\documentclass[a4paper,UKenglish,cleveref,autoref,thm-restate]{lipics-v2021}
\usepackage{times}  

\usepackage{hyperref}
\usepackage{newfloat}
\usepackage{listings}
\title{Dung's Argumentation Framework: Unveiling the Expressive Power with Inconsistent Databases}
\titlerunning{Expressivity Comparison between Dung's AFs and Inconsistent DBs}

\author{Yasir Mahmood}{DICE group, Department of Computer Science, Paderborn University, Germany}{yasir.mahmood@uni-paderborn.de}{https://orcid.org/0000-0002-5651-5391}{}

\author{Markus Hecher}{CSAIL, Massachusetts Institute of Technology, United States}{hecher@mit.edu}{https://orcid.org/0000-0003-0131-6771}{}

\author{Axel-Cyrille Ngonga Ngomo}{DICE group, Department of Computer Science, Paderborn University, Germany}{axel.ngonga@uni-paderborn.de}{https://orcid.org/0000-0001-7112-3516}{}

\authorrunning{Y.~Mahmood, M.~Hecher, and A-C.~Ngonga Ngomo}
\Copyright{Yasir Mahmood, Markus Hecher, and Axel-Cyrille Ngonga Ngomo}
\ccsdesc[100]{Computing methodologies~Knowledge representation and reasoning}

\keywords{Abstract Argumentation, Databases, Integrity Constraints, Functional and Inclusion dependencies,  Repairs, Expressive Power} 

\nolinenumbers 

\hideLIPIcs  

\usepackage{soul}
\usepackage{amsthm}
\usepackage{booktabs}
\usepackage{algorithm}
\usepackage{algorithmic}
\usepackage{tikz}
\usepackage{macros}
\usepackage{multicol}
\usepackage{bibentry}

\begin{document}

\maketitle

\begin{abstract}
	The connection between inconsistent databases and Dung's abstract argumentation framework has recently drawn growing interest.
	Specifically, an inconsistent database, involving certain types of integrity constraints such as functional and inclusion dependencies, can be viewed as an argumentation framework in Dung's setting.
	Nevertheless, no prior work has explored the exact expressive power of Dung's theory of argumentation when compared to inconsistent databases and integrity constraints.
	In this paper, we close this gap by arguing that an argumentation framework can also be viewed as an inconsistent database.
	We first establish a connection between subset-repairs for databases and extensions for AFs considering conflict-free, naive, admissible and preferred semantics.
	Further, we define a new family of attribute-based repairs based on the principle of maximal content preservation.
	The effectiveness of these repairs is then highlighted by connecting them to stable, semi-stable, and stage semantics.
	Our main contributions include translating an argumentation framework into a database together with integrity constraints. Moreover, this translation can be achieved in polynomial time, which is essential in transferring complexity results between the two formalisms.
	
\end{abstract}

\section{Introduction}
Formal argumentation serves as a widely applied framework for modeling and evaluating arguments and their reasoning, finding application in various contexts.
In particular, Dung's abstract argumentation framework~\cite{dung1995acceptability} has been specifically designed to model conflict relationships among arguments. 
An abstract argumentation framework (AF) represents arguments and their conflicts through directed graphs and allows for a convenient exploration of the conflicts at an abstract level.
The semantics for AFs is described in terms of sets of arguments (called extensions) that can be simultaneously accepted in a given framework.

A related yet distinct domain, with its primary focus on addressing inconsistent information, involves repairing knowledge bases and consistent query answering (CQA)~\cite{ArenasBC99,chomicki2007consistent,leopoldo2011database,bienvenu2013tractable}.
The goal there is to identify and \emph{repair} inconsistencies in the data and obtain a consistent knowledge base (KB) that satisfies the imposed constraints.
The current research on repairs and the theory of argumentation exhibit some overlaps.
Indeed, both (CQA for inconsistent KBs and instantiated argumentation theory) address reasoning under inconsistent information~\cite{croitoru2013}. 

While the connection between AFs and inconsistent databases has gained significant attention, no prior work 
has investigated the \emph{precise expressive power} of Dung's theory of argumentation in terms of integrity constraints (ICs).
In this paper, we take on this challenge and resolve the exact expressive power of Dung's AFs.
This is achieved by expressing AFs in terms of inconsistent databases where ICs include functional (FDs) and inclusion~dependencies (IDs).
As extensions in an AF are subsets of arguments satisfying a semantics, the association with relational databases is clarified through subset-repairs~\cite{CHOMICKI200590}.
Precisely, arguments are seen as database tuples, which together with ICs model the arguments interaction.

We \emph{complete the mutual relationship} between inconsistent databases and argumentation frameworks in Dung's setting, thereby strengthening the connection between the two formalisms as anticipated earlier 
\cite{mahmood2024computing}. 
Then, the conflict relationship between arguments closely resembles the semantics of functional dependencies, while the defense/support relation mirrors that of inclusion dependencies.
We establish that \emph{AFs can be seen as inconsistent databases}, 
as so far only the converse has already been established~\cite{bienvenu2020querying,mahmood2024computing}.
This strong connection offers a \emph{tabular} representation of the \emph{graphical} AFs and demonstrates that FDs and IDs alone suffice to encode argument interactions in an AF.



In relational databases, the prominent notions of repairs include set-based repairs~\cite{arenas1999consistent,CHOMICKI200590,tenCate:2012}, attribute-based repairs~\cite{Wijsen:2003}, and cardinality-based repairs~\cite{LopatenkoB07}.
Here, we focus on subset-repairs. 
A \emph{subset-repair of} an \emph{inconsistent database} is obtained by removing tuples from the original database such that integrity constraints are no longer violated.
It is known~\cite{CHOMICKI200590} that for large classes of constraints (FDs and denial constraints), the
restriction to deletions 
suffices to remove integrity violations.
Subset-maximality is employed in this setting to assure minimal tuple removals.

Finally, we \emph{introduce a new family} of (attribute-based) database repairs based on a pre-existing principle of maximal content preservation.
The prior research on attribute-based repairs focuses on updating  attribute values~\cite{Wijsen:2003,FlescaCQA05,bertossi2008complexity} and does not consider the setting of subset-repairs. 
We define the notion of covering repairs that maximally (fully) preserve the attribute values in a database.
In other words, repairs that encompass a larger set of values for their attributes are preferred among all subset-repairs.
We propose this novel family of repairs as a topic of independent interest.
The relevance and practical implications of these repairs are underlined by connecting them to various AF~semantics.

\paragraph*{Contributions.}
An overview of our main contributions is depicted in Table~\ref{tab:cont}. 
In details, we establish the following.
\begin{itemize}
	\item We present a database view for Dung's theory of argumentation and prove that an AF can be seen as an inconsistent database in the presence of functional and inclusion dependencies. This also establishes the exact expressive power of AFs in terms of integrity constraints. 
	\item We prove that the extensions of an AF correspond precisely to the subset-repairs of the resulting database for conflict-free, admissible, naive and preferred semantics.
	\item We propose a new family of subset-repairs based on 
	\emph{maximal content preservation}. While being of independent interest, they further tighten the connection of AFs to databases for stable, semi-stable and stage semantics. 
\end{itemize}

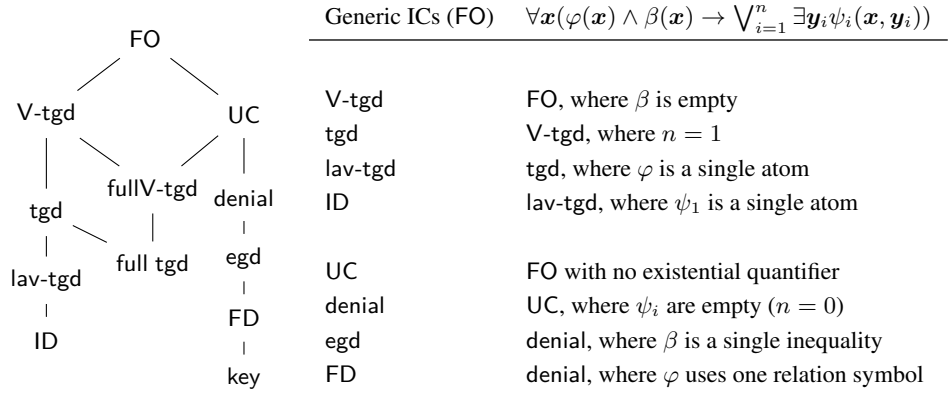
\begin{figure}
\centering
		\begin{tikzpicture}[level distance=1.7em, sibling distance=5em, minimum height=1.75em, level 2/.style={sibling distance=3em},
		every node/.style={scale=1, font=\small},
		edge from parent/.style={thick,-,black, draw}]
		\node (fo) at (0,-.2)  {$\FO$};
		\node (vtgd) at (-1.3,-1.2) {$\vtgd$};
		\node (UC) at (1.3,-1.2)  {$\UC$};
		\node (tgd) at (-1.3,-2.5) {$\tgd$};
		
		\node (fullvtgd) at (0.1,-2.2) {$\fullvtgd$};
		\node (fulltgd) at (0.1,-3.2) {$\fulltgd$};
		
		\node (lavtgd) at (-1.3,-3.4) {$\lavtgd$};
		\node (id) at (-1.3,-4.2) {$\id$};
		
		\node (denial) at (1.3,-2.3) {$\denial$};
		\node (egd) at (1.3,-3.1) {$\egd$};
		\node (fd) at (1.3,-3.9) {$\fd$};
		\node (key) at (1.3,-4.7) {$\key$};
		
		\foreach \f/\g in {fo/vtgd, fo/UC, UC/denial, denial/egd, egd/fd, fd/key, UC/fullvtgd, fullvtgd/fulltgd, vtgd/fullvtgd, tgd/fulltgd, vtgd/tgd, tgd/lavtgd,lavtgd/id} {
			\draw[-] (\f) -- (\g);
		}
	\end{tikzpicture}
	\hspace{.05cm}
	\begin{tikzpicture}[level distance=1.7em, sibling distance=5em, minimum height=1.75em, level 2/.style={sibling distance=3em},
		every node/.style={scale=1, font=\large},
		edge from parent/.style={thick,-,black, draw}]
		
		\node (uc) at (0,0)  {
			\begin{tabular}{l l}
				Generic ICs ($\FO$) & $\forall \tuple{x} (\varphi(\tuple{x})\land \beta(\tuple x)\rightarrow \bigvee_{i=1}^n \exists \tuple{y}_i \psi_i(\tuple x, \tuple {y}_i))$\\\midrule \\
				$\vtgd$   & $\FO$, where $\beta$ is empty \\ 
				$\tgd$    & $\vtgd$, where $ n=1$ \\ 
				$\lavtgd$ & $\tgd$, where $\varphi$ is a single atom \\
				$\id$ & $\lavtgd$, where $\psi_1$ is a single atom \\
				\\
				
				$\UC$     & $\FO$ with no existential quantifier \\
				$\denial$ & $\UC$, where $ \psi_i$ are empty ($n=0$)\\
				$\egd$    & $\denial$, where $\beta$ is a single inequality \\
				$\fd$ &  $\denial$, where $\varphi$ uses one relation symbol
				
		\end{tabular}};

	\end{tikzpicture}
	\caption{Hierarchy of ICs and syntactic form for most commonly studied constraints~\cite{arming2016complexity}. Formulas $\varphi$ and $\psi_i$ are conjunctions of database atoms and $\beta$ is a formula using only (in)equality symbols.}
	\label{fig:ICs}
\end{figure}

Naturally, an AF~$(A,R)$ can be seen as a database with a unary relation $A$ (arguments) and a binary relation $R$ (attacks).
%
%
Our main contributions indicate that AFs are as low in expressive power as DBs \emph{with only FDs and IDs}.
Both types of constraints are located at the lower ends of ICs hierarchy in terms of expressive power as highlighted in Figure~\ref{fig:ICs}. 
Interestingly, while the attack relation fully characterizes an AF, FDs alone can only capture this as a conflict.
Fundamentally, FDs are less expressive than the attack relation and require additional support. We provide this support by IDs, which can represent the defense relation between arguments but not the conflicts.
Our reductions \emph{broaden the applicability} of systems based on evaluating DBs with ICs, e.g., \cite{dixit2019sat,kolaitis2013efficient}.
Since conjunctive queries offer a powerful tool for analyzing databases, a DB perspective on AFs enables fine-grained reasoning.
This enables \emph{queries beyond} extension existence or credulous/skeptical reasoning, a topic which has been motivated earlier~\cite{dvovrak2012abstract}.

Resolving the expressivity of argumentation frameworks also has a wider impact.
In the argumentation community, extensions and generalizations for AFs are actively proposed and researched, such as acceptance conditions in terms of 
constraints~(Coste-Marquis et al.~\cite{Coste-MarquisDevredMarquis06a}, Alfano et al.~\cite{alfano2021argumentation})
or abstract dialectical frameworks~\cite{BrewkaWoltran10}. 
Our results, indicating that AFs have \emph{limited expressive power}, 
underline the importance of this research area.
Moreover, 
as we will show, stable, semi-stable, and stage semantics maximize certain aspects (range) of accepted arguments. 
This allows to define repairs maximizing certain attribute values for databases. 
Such repairs introduce a \emph{set-level preference} between repairs (based on data coverage), which has not been considered before. 

\paragraph*{Related Works.}
AFs have been explored extensively for reasoning with inconsistent KBs 
\cite{cayrol1995relation,vesic2012,AriouaCV17,YunVC20,bienvenu2020querying} and explaining query answers~\cite{croitoru2014query,arioua2015query,hecham2017empirical}. 
The common goal involves formally establishing a connection between KBs and AFs, thus proving the equivalence of extensions to the repairs for a KB.
This yields an \emph{argumentative view} of the inconsistent KBs and allows implementing the CQA-semantics via AFs. 
FDs and IDs are two most commonly studied ICs in databases~\cite{CHOMICKI200590,livshits2020computing}.
The translation from 
DBs to AFs are known,
proving that subset-repairs for an inconsistent database with FDs and IDs align with preferred extensions in the resulting AF~\cite{mahmood2024computing}. 
Finally, \cite{konig2022just} provides several translations between 
Assumption Based Argumentation \cite{bondarenko1997abstract}, Claim-Augmented Frameworks~\cite{dvovrak2020complexity}, and Argumentation Frameworks with Collective Attacks~\cite{NielsenP06}.
We contribute to this line of research by providing translations from AFs to DBs.

\begin{table*}[t]
	\centering
		\resizebox{.8\textwidth}{!}{%
		\begin{tabular}{c@{\; }c@{\; }c@{\; }c@{\; }c@{\; }c@{\; } c@{\; }}
			\toprule
			$\sigma $ & Repairs 	& Maximality 	&  ICs  
			& Table Size & FDs/IDs Size & Refs.
			\\ \toprule  
			$\conf$ & $\repairs$ & --- & FDs 
			& $ |A| \times (|A|+1)$ & $|A|$ & Thm.~\ref{thm:conf-naive}/\ref{thm:conf-naive-time}\\
			$\naive$ & $\repairsmax$ & Subset & FDs 
			& $ |A| \times (|A|+1)$ & $|A|$& Thm.~\ref{thm:conf-naive}/\ref{thm:conf-naive-time} \\
			$\adm$ & $\repairs$	& ---& FDs+IDs  & $|A| \times 3(|A|+1)$ & $|A|$/$(|A|+1)$ & Thm.~\ref{thm:adm-pref}/\ref{thm:AF-time}\\
			$\pref$ & $\repairsmax$	& Subset & FDs+IDs 
			& $|A| \times 3(|A|+1)$ & $|A|$/$(|A|+1)$ & Thm.~\ref{thm:adm-pref}/\ref{thm:AF-time} \\
			$\stab$ & $\repairsfullcov$	& Full-covering	& FDs+ID & $|A| \times (2|A|+3)$ & $|A|$/$1$ & Thm.~\ref{thm:stab-stag}/\ref{cor:stab-stag-time}\\
			$\stag$ & $\repairsmaxcov$	& Max-covering	& FDs+ID & $|A| \times (2|A|+3)$ & $|A|$/$1$ & Thm.~\ref{thm:stab-stag}/\ref{cor:stab-stag-time}\\
			$\semistab$ & $\repairsmaxcov$	& Max-covering & FDs+IDs & $|A| \times 3(|A|+1)$ & $|A|$/$(|A|+1)$ & Thm.~\ref{thm:stab-stag}/\ref{thm:AF-time}\\
			\bottomrule
		\end{tabular}
		}
	\caption{Overview of our main contributions. The AF-semantics~$\sigma$ (column-I) corresponds to repairs (column-II) with the type of maximality imposed (column-III), followed by integrity constraints needed to simulate AF-semantics (column-IV), the size of the resulting database tables (column-V) and the size of sets of ICs (column-VI), and references to the proofs (column-VII). The size of the database is represented as the (number of rows)$\times$(number of columns) for number of arguments $|A|$ in the AF.
	}\label{tab:cont}
\end{table*}
\section{Preliminaries}\label{sec:preli}
In the following, we briefly recall the relevant definitions. 

\smallskip
\noindent\textbf{Abstract Argumentation.}
We use Dung's argumentation framework~\cite{dung1995acceptability} and consider non-empty finite sets of arguments~$A$.
An \emph{(argumentation) framework~(AF)} is a directed graph~$\calF=(A, R)$, where $A$ is a set of arguments and the relation $R \subseteq A\times A$ represents direct attacks between arguments.
Let $S\subseteq A$, an argument~$a \in A$ is \emph{defended by $S$ in $\calF$}, if for every $(b, a) \in R$ there exists $c \in S$ such that $(c, b) \in R$.
For $a,b\in A$ such that $(b,a)\in R$, we denote by $\defend{b}{a}\dfn\{c \mid  (c,b)\in R\}$ the set of arguments defending $a$ against the attack by $b$.
The characteristic function $\adef(S)\colon 2^A \rightarrow 2^A$ of $\calF$ defined as $\adef(S) \dfn \{ a \mid a \in A, a \text{ is defended by $S$ in $\calF$}  \}$ assigns each set the arguments it defends.
The \emph{degree} of an argument $a\in A$ is the number of arguments attacking $a$ or attacked by $a$. The \emph{degree of} $\calF$ is the maximum degree of any $a\in A$.

In abstract argumentation one is interested in computing \emph{extensions}, which are subsets~$S \subseteq A$ of the arguments that have certain properties.
The set~$S$ of arguments is called \emph{conflict-free in~$\calF$} if $(S\times S) \cap R = \emptyset$.
Let $S$ be conflict-free, then $S$ is
\emph{naive in $\calF$} if no $S' \supset S$ is \emph{conflict-free} in $\calF$;
\emph{admissible in $\calF$} if every $a \in S$ is \emph{defended by $S$ in $\calF$}.
%
Further, let $S^+_R:=S\cup\{\, a\mid (b,a)\in R, b \in S\, \}$ and 
$S$ be admissible. Then, $S$ is
\emph{complete in~$\calF$} if $\adef(S) = S$;
\emph{preferred in~$\calF$}, if no $S' \supset S$ is \emph{admissible in $\calF$};
\emph{semi-stable in $\calF$} if no admissible set $S' \subseteq A$ in~$\calF$ with~$S^+_R\subsetneq (S')^+_R$ exists; and 
\emph{stable in~$\calF$} if every $b \in A \setminus S$ is \emph{attacked} by some $a \in S$.
Finally, a conflict-free set~$S$ is \emph{stage in $F$} if there is no conflict-free~$S'\subseteq A$ in~$F$ with~$S^+_R\subsetneq (S')^+_R$.
For each semantic~$\sigma \in \{\naive, \adm, \comp, \pref, \semistab, \stab,\stag\}$, $\sigma(\calF)$ denotes the set of \emph{all extensions} of semantics~$\sigma$ in $\calF$. 
%
%

%

\begin{example}\label{intro:ex-AF}
	Consider the AF~$\calF =(A,R)$ depicted in Figure~\ref{fig:ex-AF}.
	Then, $\sigma(\calF) = \{\{b,d\},\{a\}\}$ for $\sigma\in \{\pref,\stab,\comp\}$.
	Further,
	$\conf(\calF) = \pref(\calF)\cup \{\{x\} \mid x\in A\}\cup \{\emptyset\}$, 
	$\naive(\calF) = \pref(\calF) \cup \{\{c\}\}$, and
	$\adm(\calF) =  \{\{b,d\},\{b\},\{a\},\emptyset \}$.
	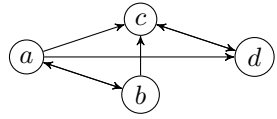
\begin{figure}[t]
		\centering
		\begin{tikzpicture}[scale=.5,arg/.style={circle,draw=black,fill=white,inner sep=.75mm}]
			\node[arg] (w) at (0,1) {$a$};
			\node[arg] (x) at (3,0) {$b$};
			\node[arg] (y) at (3,2) {$c$};
			\node[arg] (z) at (6,1) {$d$};
			\foreach \f/\t in {w/x,x/w,w/y, x/y, y/z,z/y, w/z}{
				\path[-stealth',draw=black] (\f) edge (\t) ;
			}
		\end{tikzpicture}
		\caption{Argumentation framework from Example~\ref{intro:ex-AF}.}
		\label{fig:ex-AF}
	\end{figure}
\end{example}

\noindent\textbf{Databases and Repairs.}
For our setting, an instance of a \emph{database (DB)} is a single table denoted as $T$ since it suffices to prove the connection to AFs.
Each entry in the table is called a \emph{tuple} which is associated with a unique identifier (depicted in boldface $\dbtuple{t}\in T$).
Formally, a table corresponds to a relational schema denoted as $T(x_1 ,\ldots, x_n)$, where $T$ is the relation name and $x_1,\ldots,x_n$ are distinct {attributes}.
We denote individual attributes by small letter (e.g., $x,y$) and reserve capital letters ($X,Y$) for sequences of attributes.
For an attribute $x$ and tuple $\dbtuple{s}\in T$, $\dbtuple{s}[x]$ denotes the value taken by $\dbtuple{s}$ for the attribute $x$ and for a sequence $X=(x_1,\ldots,x_k)$, $\dbtuple{s}[X]$ denotes the sequence $(\dbtuple{s}[x_1],\ldots, \dbtuple{s}[x_k])$.
For an instance $T$, $\dom(T)$ denotes the \emph{active domain} of $T$ defined as the collection of all the values occurring in any tuple in $T$.
We define the size of a table $T$ with $m$ tuples (rows) and $n$ attributes (columns) as $m\times n$.

Let $T(x_1, \ldots, x_n)$ be a schema and $T$ be a database.
In the following, we employ commonly used definitions for FDs and IDs.
A \emph{functional dependency} (FD) over $T$ is an expression of the form $\depas{X}{Y}$ 
for sequences $ X, Y$ of attributes in $T$.
A database $T$ satisfies $\depas{X}{Y}$, denoted as $T\models (\depas{X}{Y})$ if for all $\dbtuple{s},\dbtuple{t}\in T$: if $\dbtuple{s}[X]=\dbtuple{t}[X]$ then $\dbtuple{s}[Y]=\dbtuple{t}[Y]$.
That is, every pair of tuples from $T$ that agree on their values for attributes in $X$ also agree on their values for $Y$.
Moreover, an \emph{inclusion dependency} (ID) is an expression of the form $\inca{X}{Y}$ for two sequences $X$ and $Y$ of attributes with same length.
Then, $T$ satisfies $\inca{X}{Y}$ ($T\models \inca{X}{Y}$) if for each $\dbtuple{s}\in T$, there is some $\dbtuple{t}\in T$ such that $\dbtuple{s}[X]=\dbtuple{t}[Y]$.
Let $i\dfn \incas{X}{Y}\in I$ be an ID and $\dbtuple{s}\in T$, we say that a tuple $\dbtuple{t}\in T$ \emph{supports} $\dbtuple{s}$ for the ID $i$ if $\dbtuple{s}[X]=\dbtuple{t}[Y]$. 
We denote by~$\support{i}{\dbtuple{s}}\dfn\{\dbtuple{t} \mid \dbtuple{t}[Y]=\dbtuple{s}[X] \}$ the collection of tuples \emph{supporting} $\dbtuple{s}$ for $i$.


Let $T$ be a database and $B$ be a collection of FDs and IDs.
Then $T$ is \emph{consistent} with respect to $B$, denoted as $T\models B$, if $T\models b$ for each $b\in B$.
Further, $T$ is \emph{inconsistent} with respect to $B$ if there is some $b\in B$  such that $T\not \models b$.
A \emph{subset-repair} of $T$ with respect to $B$ is a subset $P\subseteq T$ which is consistent with respect to $B$. 
Moreover, $P$ is a \emph{maximal repair}\footnote{The standard repair definition insists on subset-maximality. 
	Our relaxed version clarifies the link to AF-semantics without requiring maximality. Nevertheless, one can substitute (our) \emph{repairs} with \emph{sub-repairs} to reserve \emph{repairs} for \emph{subset-maximal repairs}.} if there is no set $P'\subseteq T$ such that $P'$ is also consistent with respect to $B$ and $P \subset P'$.
We simply speak of a repair when we intend to mean a subset-repair and often consider a database $T$ without explicitly writing its schema.
Let $\calD= \langle T,D\rangle $ where $T$ is a database and $D$ is a set of dependencies, then $\repairs(\calD)$ (resp., $\repairsmax(\calD)$) denotes the set of all (maximal) repairs for $\calD$.
Slightly abusing the notation,  we call 
a table $T$ and an instance $\langle T,D\rangle$, a database.

\begin{example}\label{ex:intro-rep}
	Consider $\calD = \langle T,D \rangle$ with database $T=\{\dbtuple{s}_i \mid i\leq 6\}$ as depicted in Table~\ref{tab:intro-rep} with FD $f\dfn  \depas{\texttt{Tutor},\texttt{Time}}{\texttt{Room}}$, and ID $ \incas{\texttt{Advisor}}{\texttt{Tutor}}$. 
	Informally, a tutor in any time slot can be in at most one room  (FD) and an advisor for each course must be a tutor (ID).
	$\calD$ is inconsistent since $\{\dbtuple{s}_1,\dbtuple{s}_2\}\not\models f, \{\dbtuple{s}_3,\dbtuple{s}_4\}\not\models f$ and there is no $\dbtuple{s}_i\in T$ with  $\dbtuple{s}_6[\texttt{advisor}]=\dbtuple{s}_i[\texttt{tutor}]$.
	Then, 
	a subset containing exactly one tuple from each set $\{\dbtuple{s}_1,\dbtuple{s}_2\}$, $\{\dbtuple{s}_3,\dbtuple{s}_4\}$ and $\{\dbtuple{s}_5\}$ is a maximal repair for $\calD$.
	Further, $\{\dbtuple{s}_1\},\{\dbtuple{s}_4\},\{\dbtuple{s}_1,\dbtuple{s}_3\}$ are  repairs but not maximal.
	\begin{table}
		\centering
		\begin{tabular}{c@{\; }c@{\; }c@{\; }c@{\; }c@{\; }c@{\; }}\toprule
			$T$ & \texttt{Tutor} 	& 	\texttt{Time}		& 	\texttt{Room} 	& 	\texttt{Course} & \texttt{Advisor} \\ \toprule
			$\dbtuple{s}_1$ & Alice	& TS-1	& A10	& Logic-I	& Alice 	\\
			$\dbtuple{s}_2$ & Alice	& TS-1	& B20 	& Algorithms	& Carol  	\\
			$\dbtuple{s}_3$ & Bob	& TS-2	& B20 	& Statistics	& Alice 	\\
			$\dbtuple{s}_4$ & Bob	& TS-2	& C30 	& Calculus	& Bob 	 \\
			$\dbtuple{s}_5$ & Carol	& TS-3	& C30 	& Calculus	& Bob 	 \\
			$\dbtuple{s}_6$ & Carol	& TS-3	& B20 	& Algorithms	& Dave \\
			\bottomrule
		\end{tabular}
		\caption{An inconsistent database.} 
		\label{tab:intro-rep}
	\end{table}
\end{example}
%

%

\section{A DB View of Abstract Argumentation}

In this section, we connect inconsistent databases and abstract argumentation frameworks. 
This is established by proving that an AF~$\calF$ can be seen as an instance $\AF{D}{F}$ of inconsistent database such that the acceptable sets of arguments in $\calF$ correspond precisely to (maximal) repairs for $\calD$.
%
Intuitively, our construction relies on the fact that the attack relation between arguments in $\calF$ can be simulated via a database together with FDs, and defending arguments can be seen as another database with IDs.
Then, combining the two databases and dependencies, we obtain an AF --- seen as an inconsistent database.
For convenience, an argument ``$a$'' seen as a tuple in the database is denoted by ``$\dbtuple{a}$'' to highlight that $\dbtuple{a}$ depicts a tabular representation of $a$.

We first simulate the conflicts between arguments in an AF $\calF$ via FDs, resulting in a \emph{conflict database}.
Then, 
we establish that the idea of \emph{defending} arguments can be simulated by a \emph{defense} database and IDs.
Finally, 
we combine the two databases thereby proving that in AFs, the semantics $\sigma\in \{\conf,\naive,\adm,\pref\}$ can be simulated via (maximal) repairs for inconsistent databases with FDs and IDs.
In addition, we establish the size and the time complexity associated with the construction of the aforementioned databases. 
For simplicity, we first assume AFs without self-attacking arguments. 
At the end, we highlight the changes required to allow self-attacking arguments in AFs.

\subsection{Conflict DBs via Functional Dependencies}\label{sec:attack}
Let $\calF =(A,R)$ be an AF with a set~$A$ of arguments and attacks $R$.
Then, we construct a \emph{conflict} database $\AF{C}{F} = \langle T,F\rangle $ 
where 
each argument $a\in A$ is seen as the tuple $\dbtuple{a}$ (the unique identifier to the tuple representing the argument $a$) and a collection $F = \{f_i \mid r_i\in R\}$ of FDs.
A tuple $\dbtuple{a}\in T$ encodes the 
conflicts between $a\in A$ and other arguments in $\calF$.
Intuitively, $\calC$ is constructed in such a way that $\{\dbtuple{a},\dbtuple{b}\}\not \models f_i$ for the attack $r_i =(a,b) \in R$.

Notice that, $\{\dbtuple{a},\dbtuple{b}\} =  \{\dbtuple{b},\dbtuple{a}\}$
is true for any $\dbtuple{a},\dbtuple{b}\in T$, although the attacks $(a,b)$ and $(b,a)$ are not the same.
As a consequence, the conflict database for an AF will be \emph{symmetric} by design.
That is, 
both attacks $(a,b),(b,a)\in R$ are encoded in the conflict database in the same way by requiring that $\{\dbtuple{a},\dbtuple{b}\}\not\models f$ for some $f\in F$.
It is worth remarking that functional dependencies can only simulate the symmetric attacks.
In particular, the input AF may contain non-symmetric attacks though our translation (for the conflict database) will treat them as symmetric
ones. 
This is unproblematic in this particular setting, since only conflict-free (naive) extensions correspond to repairs (maximal repairs) with FDs. 
Further, we will expound on this situation later  that one can (perhaps) not simulate directed attacks by FDs alone, and requires IDs for that purpose.

In the following, each attack $r\in \{(a,b),(b,a)\}$ is depicted as a set 
$r\dfn \{a,b\}$.

\begin{definition}[Conflict Database]\label{db:attack}
	Let $\calF =(A,R)$ be an AF with arguments $A$ and attacks $R$.
	Then, $\AF{C}{F}\dfn \langle T,F\rangle$ defines the \emph{conflict database} for $\calF$, specified as follows.
	\begin{itemize}
		\item The attributes of $T$ are $\{x_i,n \mid r_i\in R\}$ and we call $n$ the \emph{name} attribute.
		\item $F\dfn \{\depas{x_i}{n} \mid r_i\in R\}$ is the collection of FDs and each $f\in F$ \emph{corresponds} to an attack $r\in R$.
		\item $T\dfn \{\dbtuple{a} \mid a\in A\}$. Further, 
			(1) $\dbtuple{a}[n] =a$ for each $a\in A$, (2) for $r_i=\{a,b\}$: $ \dbtuple{a}[x_i] = r_i = \dbtuple{b}[x_i]$, and 
			(2) for $\dbtuple{c}\in T$ and attributes $x_j$ not already assigned (when all the attacks have been considered): $\dbtuple{c}[x_j] =c$.
	\end{itemize}
\end{definition}
%
%

The following example presents the conflict database for the AF~$\calF =(A,R)$ from our running example (Ex.~\ref{intro:ex-AF}). 
\begin{example}\label{ex:conf-FD}
	For brevity, we rename each attack in $\calF$ as depicted in Fig.~\ref{fig:conf-FD}.
	As discussed, the attacks $(a,b),(b,a)\in R$ are modeled as one conflict $r_1\dfn \{a,b\}$.
	The conflict database for $\calF$ is $\AF{C}{F}=\langle T,F\rangle$, where $F=\{\depas{x_i}{n} \mid i\leq 5\}$ and $T$ as specified in Figure~\ref{fig:conf-FD}.
	
	\begin{figure}[t]
		\centering
		\begin{tikzpicture}[scale=.52,arg/.style={circle,draw=black,fill=white,inner sep=.75mm}]
			\node[arg] (w) at (0,1) {$a$};
			\node[arg] (x) at (3,0) {$b$};
			\node[arg] (y) at (3,2) {$c$};
			\node[arg] (z) at (6,1) {$d$};
			\foreach \f/\t/\u in {x/w/,w/y/$r_2$, y/z/$r_4$,z/y/}{
				\path[-stealth',draw=black] (\f) edge node[above] {\u} (\t) ;
			}
			\path[-stealth',draw=black] (w) edge node[below, near start] {$r_1$} (x) ;
			\path[-stealth',draw=black] (w) edge node[below, near end] {$r_5$} (z) ;
			\path[-stealth',draw=black] (x) edge node[left, near start] {$r_3$} (y) ;
			\node (table) at (11.5,1) {
				\begin{tabular}{l @{\;}| @{\;}cccccc}
					$T$ & $x_1$  & $x_2$  & $x_3$ & $x_4$  & $x_5$ & $n$ \\\hline 
					$\dbtuple{a}$ &$r_1$ &$r_2$ & $a$ & $a$ & $r_5$ & $a$ \\
					$\dbtuple{b}$ &$r_1$ & $b$ &$r_3$ & $b$ &$b$ &$b$  \\
					$\dbtuple{c}$ & $c$ & $r_2$ &$r_3$ & $r_4$ &$c$& $c$  \\				
					$\dbtuple{d}$ & $d$ & $d$ &$d$ &$r_4$ & $r_5$ & $d$ \\
			\end{tabular}};
		\end{tikzpicture}
		\caption{The AF $\calF$ (left) from Ex.~\ref{ex:conf-FD} and a conflict database (right) modeling $\calF$.}
		\label{fig:conf-FD}
	\end{figure}
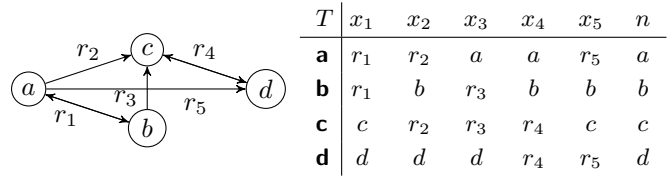
\end{example}

%
%
Let $S\subseteq A$, then by $S_T \subseteq T$ we denote the set of tuples corresponding to arguments in $S$, defined as $S_T \dfn \{\dbtuple{a} \mid \dbtuple{a}\in T, a\in S\}$.
The following theorem establishes the relation between the extensions and repairs of the conflict database.
\begin{theorem}\label{thm:conf-naive}
	Let $\calF$ be an AF without self-attacking arguments and $\AF{C}{F}$ be its corresponding conflict database. Then, for every $S\subseteq A$, $S$ is conflict-free (resp., naive) in $\calF$ iff $S_T \subseteq T$ is a repair (maximal) for $\AF{C}{F}$.
\end{theorem}

\begin{proof}
	We prove the equivalence between conflict-free extensions of $\calF$ and subset-repairs for $\AF{C}{F}$. 
	Then the claim for naive extensions and maximal repairs follows analogously.
	
	Let $S\subseteq A$ be conflict-free, then $r\not\subseteq S$ for every $r \in R$.
	Notice that, two tuples $\dbtuple{a},\dbtuple{b}\in T$ share values for some attribute $x$ only if there is an attack between the corresponding arguments (i.e., there is  some $r\in R$, with $r=\{a,b\}$).
	This implies that $\dbtuple{a}[x]\neq \dbtuple{b}[x]$ for each attribute $x$ in $T$ and distinct $\dbtuple{a},\dbtuple{b}\in S_T$.
	Consequently, $S_T\models f$ trivially for each $f\in F$ 
	and therefore $S_T$ is a repair for $\AF{C}{F}$.
	Conversely, if $S$ is not conflict-free then $r\subseteq S$ for some $r\in R$. 
	Let $r=\{a,b\}$ and $a,b\in S$ be distinct arguments (which are guaranteed since $\calF$ does not contain self-attacking arguments).
	Then there exists an attribute $x$, such that, $\dbtuple{a}[x]=r=\dbtuple{b}[x]$ but $\dbtuple{a}[n]=a\neq b = \dbtuple{b}[n]$ for $\dbtuple{a},\dbtuple{b}\in S_T$. 
	Now, $S_T$ can not be a repair for $\AF{C}{F}$ since $S_T\not\models \depas{x}{n}$. This completes the proof.
\end{proof}

Notice that the above construction fails in general when $\calF$ contains self-attacking arguments.
That is, if $ (a,a)\in R$, the only argument participating in a conflict is $a$ but the database consisting of a singleton tuple $\{\dbtuple{a}\}$ satisfies each FD trivially.
We will see that this issue can be resolved by allowing an ID since a singleton tuple can also fail IDs~(see Example~\ref{ex:intro-rep}).

%

%

\paragraph*{The sparsity issue for conflict databases.}
Observe that the construction of conflict database from an AF includes an individual attribute $x_r$ modeling $r\in R$.
This translation has the side-effect that conflict databases are very sparse in the following sense: for each attribute $x$ (except the name attribute $n$), there are exactly two tuples $\dbtuple{a},\dbtuple{b}\in T$ (corresponding to the attack $r=\{a,b\}$) having \emph{meaningful} values for $x$, whereas $\dbtuple{c}[x]=c$ for each $\dbtuple{c}\in T\setminus \{\dbtuple{a},\dbtuple{b}\}$.
This results in a very large database with most values being repeated when the number of attacks in $\calF$ is also large.
We argue that this sparsity of the conflict database can be avoided by \emph{reusing} FDs in $F$.
Intuitively, a single attribute $x$ and FD $f\in F$ can be used for two different attacks that do not share an argument. 
That is, an attribute $x$ and FD $\depas{x}{n}$ suffices to model each conflict $r_i= \{a_i,b_i\}\in R$ 
between distinct arguments $a_i\neq a_j\neq b_i\neq b_j$ for each $i\neq j$. 
This works by setting $\dbtuple{a}_i[x] {\,=\,}r_i {\,=\,}  \dbtuple{b}_i[x]$, $\dbtuple{a}_i[n] {\,=\,} a_i$ and $ \dbtuple{b}_i[n] {\,=\,} b_i$ for each $i$.
%

We prove that the problem of assigning attributes and FDs to conflicts in $R$ can be equivalently seen as the edge-coloring problem.
That is, given a graph $\calG =(V,E)$, determine whether there exists a coloring for edges in $E$ such that no two adjacent edges are assigned the same color.
The following lemma makes this connection precise, thereby allowing us to deduce that it actually suffices to use the number of FDs determined by the maximum degree of input AF.

\begin{lemma}\label{lem:few-FDs}
	Let $\calF$ be an AF of degree $\gamma$.
	Then, there is a conflict database $\AF{C}{F}\dfn \langle T, F\rangle$ for $\calF$ such that $|F|\leq \gamma{+}1$. 
\end{lemma}

\begin{proof}
Let~$\calF=(A,R)$ be an AF with arguments $A$ and conflicts $R$. 
%
Assume $a_i, a_j$ and $b_i, b_j$ be distinct arguments such that $r_i\dfn \{a_i,b_i\}\in R$ and $r_j\dfn \{a_j,b_j\}\in R$.
Moreover, suppose $f_i\dfn x_i\rightarrow n$ and $f_j\dfn x_j\rightarrow n$ denote the FDs encoding the conflicts $r_i$ and $r_j$.
Let $x_k$ be a fresh attribute, then we denote by $\AF{C}{F}'=\langle T',F'\rangle$ the database instance after replacing attributes $x_i,x_j$ (resp., FDs $f_i,f_j$ in $F$) by $x_{k}$ ($f_{k}$).
Further, in $\AF{C}{F}'$ we set $\dbtuple{a}_i[x_{k}]=r_i=\dbtuple{b}_i[x_{k}]$ and $\dbtuple{a}_j[x_{k}]=r_j=\dbtuple{b}_j[x_{k}]$, whereas the remaining values remain unchanged.
Then, the following claim depicts the relationship between each $P\subseteq T$ in $\AF{C}{F}$ versus in $\AF{C}{F}'$.
\begin{claim}\label{claim:sat}
	$P\models \{f_i,f_j\}$ in $\AF{C}{F}$ iff $P\models f_k$ in $\AF{C}{F}'$.
\end{claim}
\noindent\textbf{Proof of Claim.}
Observe that the attribute $x_k$ in $\AF{C}{F}'$ and $x_i,x_j$ in $\AF{C}{F}$ for tuples $\{\dbtuple{a}_i,\dbtuple{b}_i,\dbtuple{a}_j,\dbtuple{b}_j\}$ include the only changes between $\AF{C}{F}'$ and $\AF{C}{F}$.
Then, 
$P\models \{f_i,f_j\}$ iff $P$ contains at most one tuple from each set $\{\dbtuple{a}_i,\dbtuple{b}_i\}$ and $\{\dbtuple{a}_j,\dbtuple{b}_j\}$.
Moreover, the attribute $x_k$ is assigned value $r_i$ by $\{\dbtuple{a}_i,\dbtuple{b}_i\} $ and $r_j$ by $\{\dbtuple{a}_j,\dbtuple{b}_j\}$.
As a result, $P\models f_k$ in $\AF{C}{F}'$ since other attributes in $\AF{C}{F}'$ and $\AF{C}{F}$ remain the same. 
Conversely, 
if $P\models f_k$, then once again $P$ contains at most one tuple from each set $\{\dbtuple{a}_i,\dbtuple{b}_i\}$ and $\{\dbtuple{a}_j,\dbtuple{b}_j\}$. Further, $\dbtuple{a}_i[x_j]=a_i\neq \dbtuple{b}_i[x_j]=b_i$ and the same is true for tuples $\{\dbtuple{a}_j,\dbtuple{b}_j\}$ and attribute $x_i$.
This implies that $P\models \{f_i,f_j\}$ and establishes the claim. \qedclaim

Now, we are ready to prove that the task of assigning attributes and FDs to conflicts in $R$ can be equivalently seen as the edge-coloring problem for $\calF$.
Indeed, given a valid edge coloring, we use an attribute $x_c$ for each color $c$ and FD  $x_c\rightarrow n$. 
In other words, distinct colors represent attributes and a valid coloring implies that no two edges connected to a vertex share the same attribute.
This results in a well-defined 
assignment of attribute values via tuples by setting $\dbtuple{a}[x_c]=r_i = \dbtuple{b}[x_c]$ for every conflict $r_i=\{a,b\}$ with color $c$ in $\calF$.
Further, observe that each FD in $\AF{C}{F}$ can be written as $\depas{x_c}{n}$ for some color $c$ since every edge is colored.
Finally, as  Claim~\ref{claim:sat} proves, the change of attribute values and FDs concerning only distinct conflicts (not sharing any end point) preserves the satisfaction of FDs in every subset $P\subseteq T$. 
Consequently, repairs for $\AF{C}{F}$ and $\AF{C}{F}'$ coincide where $\AF{C}{F}'$ is the conflict database in which FDs and attributes are assigned by an edge-coloring of $\calF$.

The equivalence between two problems (assigning attributes/FDs to conflicts in $\calF$ and edge-coloring $\calF$) allows us to utilize the well-known Vizing's theorem~\cite{vizing1964estimate}, stating that the number of colors needed to edge-color a graph is at most $\gamma+1$ for a graph  of degree $\gamma$.
\end{proof}
The proof of Lemma~\ref{lem:few-FDs} also allows us to deduce that the conflict database constructed with fewer FDs still obeys the equivalence between (naive) conflict-free extensions in $\calF$ and (maximal) repairs for $\AF{C}{F}$~(Theorem~\ref{thm:conf-naive}). 

\begin{example}\label{main-ex:conf-FD}
Figure~\ref{main-fig:conf-FD} depicts a compact representation of the conflict database from our running  example. 
Clearly, the colors for the edge-coloring of $\calF$ result in three attributes $\{x_b,x_r,x_g\}$ and FDs $\depas{x_i}{n}$ for $i\in \{b,r,g\}$.
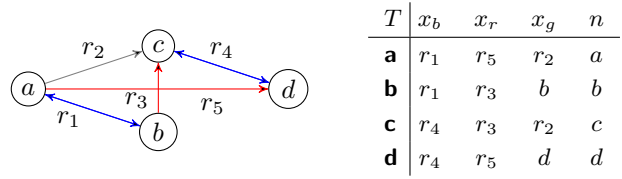
\begin{figure}[t]
	\centering
	\begin{tikzpicture}[scale=.57,arg/.style={circle,draw=black,fill=white,inner sep=.75mm}]
		\node[arg] (w) at (0,1) {$a$};
		\node[arg] (x) at (3,0) {$b$};
		\node[arg] (y) at (3,2) {$c$};
		\node[arg] (z) at (6,1) {$d$};
		\foreach \f/\t/\u/\c  in {x/w//blue/,w/y/$r_2$/gray, y/z/$r_4$/blue,z/y//blue}{
			\path[-stealth',draw=\c] (\f) edge node[above] {\u} (\t) ;
		}
		\path[-stealth',draw=blue] (w) edge node[below, near start] {$r_1$} (x) ;
		\path[-stealth',draw=red] (w) edge node[below, near end] {$r_5$} (z) ;
		\path[-stealth',draw=red] (x) edge node[left, near start] {$r_3$} (y) ;
		\node (table) at (10.75,1) {
			\begin{tabular}{l @{\;}| @{\;}cccc}
				$T$ & $x_b$ & $x_r$ & $x_g$ & $n$ \\\hline 
				$\dbtuple{a}$ &$r_1$ & $r_5$ &$r_2$ & $a$  \\
				$\dbtuple{b}$ &$r_1$ & $r_3$ &$b$ & $b$  \\
				$\dbtuple{c}$ & $r_4$ & $r_3$ &$r_2$ & $c$ \\
				$\dbtuple{d}$ & $r_4$ & $r_5$ &$d$ &$d$ \\
		\end{tabular}};%
	\end{tikzpicture}
	\caption{The edge-colored AF $\calF$ (left) from Example~\ref{main-ex:conf-FD} and a conflict database $\AF{C}{F}$ with attributes assignment corresponding to the colors (right) for conflicts in $\calF$.}\label{main-fig:conf-FD}
\end{figure}
\end{example}

As a consequence of Lemma~\ref{lem:few-FDs}, the following claim holds. 

\begin{theorem}\label{thm:conf-naive-time}
Let $\calF$ be an AF without self-attacking arguments.
Then, its corresponding conflict database $\AF{C}{F}$ can be constructed in polynomial time in the size of $\calF$.
Additionally, $\AF{C}{F}$ has a table of size $|A|\times (|A|+1)$, and uses $|A|$-many FDs at most, for an AF with $|A|$ arguments. 
\end{theorem}

\begin{proof}
	Let $\calF =(A,R)$ be an AF and $\AF{C}{F}=\langle T,F\rangle$ be the corresponding conflict database.
	Recall that $T$ has exactly one tuple ${\dbtuple{a}}$ for each argument $a\in A$.
	Therefore, the number of rows (tuples) in $T$ is the same as $|A|$.
	Moreover, due to Lemma~\ref{lem:few-FDs}, it is easy to observe that the number of columns (attributes) in $\AF{C}{F}$ is also bounded by $|A|$ since $\gamma< |A|$ (and there are no self-attacks in $R$) where $\gamma$ is the maximum degree of any argument in $\calF$.
	As a result, the table in the database $\AF{C}{F}$ has at most $m$ rows and $m+1$ columns (including the name attribute ``$n$''), and therefore has the size $|A|\times (|A|+1)$.
	Moreover, there is one FD ($\depas{x_i}{n}$) corresponding to each attribute $x_i$, which proves the claim regarding the size of FDs.
	Finally, using the equivalence of assigning attributes and FDs to conflicts and the problem of edge-coloring $\calF$ (the proof of Lemma~\ref{lem:few-FDs}), one can use any polynomial time algorithm (such as, the one by~\cite{MISRAG}) for assigning the conflicts in $\calF$ to $(\gamma+1)$-many attributes, where $\gamma$ is the maximum degree of $\calF$.
	This completes the proof to our theorem.
\end{proof}

\subsection{Defense DBs via Inclusion Dependencies}\label{sec:defend}

Assume an AF $\calF =(A,R)$ with arguments $A$ and attacks $R$.
We construct a \emph{defense} database $\AF{D}{F} = \langle T, I\rangle$ for $\calF$.
Intuitively, the defense database $T $ contains the information about incoming and outgoing attacks for each argument $a\in A$. 
The attributes of $T$ are $\{u_a,v_a \mid a\in A\}$, further $T\dfn \{\dbtuple{a}\mid a\in A\}$ and the tuple $\dbtuple{a}\in T$ encodes the neighborhood of the argument $a\in A$.
Formally, $T$ is defined as follows.

\begin{definition}[Defense database]\label{db:defense}
Let $\calF =(A,R)$ be an AF with arguments $A$ and attacks $R$.
Then $\AF{D}{F}\dfn \langle T,I\rangle$ defines the \emph{defense} database for $\calF$, specified in the following.
\begin{itemize}
\item The attributes of $T$ are $\{u_a,v_a \mid a\in A\}$.
\item $I\dfn \{\incas{u_a}{v_a} \mid a\in A\}$ is the collection of IDs.
\item $T\dfn \{\dbtuple{a}\mid a\in A\}$, where
	(1) for each $r = (a,b)\in R$: 
	\begin{align*}
		& \dbtuple{a}[u_b]= b, & \dbtuple{a}[v_b]= b, \\
		& \dbtuple{b}[u_a] = a, & \dbtuple{b}[v_a] = 0.
	\end{align*}
	(2) for each $\dbtuple{t} \in T$ and variables $u_d, v_d$ not yet assigned: 
	\begin{align*}
		& \dbtuple{t} [u_d] = 0,  &  \dbtuple{t}[v_d]= 0.
	\end{align*} 
\end{itemize}
\end{definition}
Intuitively, for each $\dbtuple{a}\in T$: 
$\dbtuple{a}[v_b]$ encodes whether the argument $a$ attacks $b$, by setting $\dbtuple{a}[v_b]=b$ when $(a,b)\in R$.
Moreover, $\dbtuple{a}[u_b]$ aims at encoding whether $a$ \emph{interacts} with (either attacks, or attacked by) $b$. 
This is  
achieved by setting $\dbtuple{a}[u_b]=b$ if $(b,a)\in R$ or $(a,b)\in R$.
The intuition is to simulate attacks from an argument $a$ via an ID $\incas{u_a}{v_a}$.
The idea is that if an argument $b$ is attacked by $a$, then $b$ interacts with $a$ and hence $\dbtuple{b}[u_a]=a$.
Now, there are two ways $b$ can be defended  against the attack by $a$: either $b$ defends itself by attacking $a$ (in which case $\dbtuple{b}[v_b]=a$, hence $\{\dbtuple{b}\}\models\incas{u_a}{v_a}$), or there is an argument $c\not\in\{a,b\}$ attacking $a$ (therefore $\dbtuple{c}[v_a]=a$, and $\{\dbtuple{b},\dbtuple{c}\}\models\incas{u_a}{v_a}$ as $\dbtuple{c}[u_a]=\dbtuple{c}[v_a]=a$).

The following example demonstrates the defense database for the AF~$\calF =(A,R)$ from our running example (Ex.~\ref{intro:ex-AF}). 
\begin{example}\label{ex:def-ID}
The defense database for $\calF$ from Example~\ref{intro:ex-AF} is $\AF{D}{F}=\langle T,I\rangle$, where $I=\{\incas{u_s}{v_s} \mid s\in A\}$ and $T$ is specified in Table~\ref{fig:def-ID}.

\begin{table}
\centering
\begin{tabular}{l @{\;}| @{\;}cccccccc}
	$T$ & $u_a$ & $v_a$ & $u_b$  & $v_b$ & $u_c$ & $v_c$ & $u_d$ & $v_d$\\
	\hline 
	$\dbtuple{a}$ &$0$ & $0$ &$b$ & $b$ & $c$ & $c$ & $d$ & $d$ \\
	$\dbtuple{b}$ &$a$ & $a$ &$0$ & $0$ &$c$ &$c$ & $0$ & $0$   \\
	$\dbtuple{c}$ & $a$ & $0$ &$b$ & $0$ &$0$& $0$ & $d$& $d$ \\
	$\dbtuple{d}$ & $a$ & $0$ &$0$ &$0$ & $c$ & $c$ & $0$&  $0$   \\
\end{tabular}
\caption{The defense database for the AF $\calF$ in Example~\ref{ex:def-ID}.}
\label{fig:def-ID}
\end{table}

\end{example}
Observe that the active domain of $T$ consists of the arguments in $\calF$, as well as, an auxiliary element $0$ to serve the purpose of missing values in the database (i.e., when a value is missing, we prefer writing $0$ rather than leaving it blank).
Moreover, repeating the argument names instead of $0$ as for the case of the conflict database has an undesired effect.
Consider Example~\ref{ex:def-ID}, if we set $\dbtuple{a}[u_a]=\dbtuple{a}[v_a]=a$ instead of $0$, then $\dbtuple{a}\in \support{i}{\dbtuple{d}}$ and $\dbtuple{a}\in \support{i}{\dbtuple{c}}$ where $i=\incas{u_a}{v_a}$. 
However, the argument $a$ does not defend $c$ or $d$ in $\calF$ against the attack by $a\in A$.
Indeed, we aim at proving that the support relation between tuples in DBs for IDs is essentially the same as the defense relation between arguments.
This connection is clarified in the following lemma.

\begin{lemma}\label{lem:def-sup}
Let $a,b\in A$ be two arguments such that $(a,b)\in R$.
Then, for $\dbtuple{b}\in T$ (the tuple corresponding to $b$) and $i\dfn \incas{u_a}{v_a}\in I$, we have $\defend{a}{b} = \support{i}{\dbtuple{b}}$. 
\end{lemma}

\begin{proof}
	Let $(a,b)\in R$, then we have $\dbtuple{b}[u_a]=a$.
	So, $\support{i}{\dbtuple{b}}\dfn\{\dbtuple{c} \mid \dbtuple{c}[v_a] = \dbtuple{b}[u_a]\}$ includes all the tuples supporting $\dbtuple{b}$ for the ID $i$. 
	Clearly, for each tuple $\dbtuple{c}\in T$ such that $\dbtuple{c}\in \support{i}{\dbtuple{b}}$, we have that $(c,a)\in R$ since $\dbtuple{c}[v_a]=a$.
	Similarly, for each $c\in A$ such that $c\in \defend{a}{b}$, we know that $(c,a)\in R$. 
	As a result, $\dbtuple{b}[u_a]=a$ together with $\dbtuple{c}[v_a]=a$ implies that $\dbtuple{c}\in \support{i}{\dbtuple{b}}$. 
	This completes the proof. 
\end{proof}


Let $S\subseteq A$ be a set of arguments and $S_T\subseteq T$ be the corresponding set of tuples.
The defense database has the property that $S_T\models \incas{u_a}{v_a}$ iff each argument in $S $ is defended against the argument $a\in A$.
%
We prove that a set $S\subseteq A$ defends itself in $\calF$ iff  $S_T \subseteq T$ satisfies every ID in $I$.

\begin{lemma}\label{lem:defend-ids}
Let $\calF$ be an AF and $\AF{D}{F}$ denotes its corresponding defense database. Then, for every $S\subseteq A$, $S$ defends itself in $\calF$ iff $S_T \subseteq T$ is a repair for $\AF{D}{F}$.
\end{lemma}

\begin{proof}
	Let $S\subseteq A$ such that $S$ defends itself in $\calF$. 
	Then, for each $a\in A$ such that $(a,b)\in R$ for some $b\in S$, there is $c\in S$ such that $(c,a)\in R$.
	As a result, $\defend{a}{b}\cap S \neq \emptyset$.
	This implies that, for each $i\dfn \incas{u_a}{v_a}\in I$, $\support{i}{\dbtuple{b}}\cap S_T\neq \emptyset$ (due to  Lemma~\ref{lem:def-sup}), and  
	consequently, $S_T\models i$.
	Conversely, suppose there is some $c\in S$ such that $(d,c)\in R$ for some $d\in A$ but no $c'\in S$ attacks $d$.
	Now, $\defend{d}{c} \cap S = \emptyset$ and therefore, $\support{j}{\dbtuple{c}} \cap S_T = \emptyset$ where $j=\incas{u_d}{v_d}$ (again Lem.~\ref{lem:def-sup}).
	Then, $S_T\not\models \incas{u_d}{v_d}$; therefore $S_T$ can not be a repair in~$\AF{D}{F}$. 
\end{proof}

Note that repairs for a defense database do not yield conflict-free sets. 
In fact, we will prove later that 
one can not model conflict-freeness via inclusion dependencies alone.

\begin{example}\label{ex:def-rep}
Consider the AF~$\calF$ and its defense database $\AF{D}{F}$ from Example~\ref{ex:def-ID}. 
Then, $\{\dbtuple{a}\},\{\dbtuple{b}\}\in \repairs(\AF{D}{F})$ and each set defends itself whereas $\{\dbtuple{c}\},\{\dbtuple{d}\}\not\in \repairs(\AF{D}{F})$ and do not defend themselves. 
Further, $\{\dbtuple{a},\dbtuple{b},\dbtuple{c}\}\in \repairs(\AF{D}{F})$ and defends itself, although it is not conflict-free in $\calF$.
\end{example}


The following theorem establishes the size and the time bounds to construct a of defense database for a given AF.

\begin{theorem}\label{thm:def-time}
Let $\calF$ be an AF without self-attacking arguments.
Then, its corresponding defense database $\AF{D}{F}$ can be constructed in polynomial time in the size of $\calF$. 
Additionally, $\AF{D}{F}$ has a table of size $|A|\times 2|A|$ and uses $|A|$-many IDS, for an AF with $|A|$ arguments. 
\end{theorem}

\begin{proof}
	Let $\calF =(A,R)$ be an AF and $\AF{D}{F}=\langle T,I\rangle$ be the corresponding defense database.
	As before, $T$ has exactly one tuple ${\dbtuple{a}}$ for each argument $a\in A$.
	Therefore, the number of rows (tuples) in $T$ is the same as $|A|$.
	Moreover, $T$ has $2m$ columns (attributes) where $m=|A|$ (follows from Def.~\ref{db:defense}). 
	As a result, the table in the database $\AF{D}{F}$ has $m$ rows and $2m$ columns, and therefore has the size $|A|\times 2|A|$.
	Further, there is one ID ($\incas{u_a}{v_a}$) corresponding to each $a\in A$, which proves the claim for the size of IDs.
	This establishes the mentioned bounds on the size of the defense database.
	Finally, the claim regarding the runtime follows easily, since one only needs to determine all the incoming and outgoing attacks for each argument $a\in A$.
\end{proof}

\subsection{Inconsistent Databases for AFs}\label{sec:AFs}
We combine the (conflict and defense) databases from the previous two subsections and establish that a collection of FDs and IDs suffices to encode the entire AF.
Let $\calF =(A,R)$ be an AF.
We construct an \emph{AF-database} as the instance $\AF{A}{F} =\langle T, D\rangle$, where $T$ is the database obtained by \emph{combining} the conflict and defense databases, and $D=F\cup I$ consists of a collection of FDs and IDs.
Specifically, $\AF{A}{F}$ has the following components.
\begin{itemize}
	%
	\item $F \dfn \{\depas{x_i}{n} \mid r_i\in R\}$ and $I\dfn \{\incas{u_a}{v_a} \mid a\in A\}$.
	\item $T\dfn \{\dbtuple{a}\mid a\in A\}$ is a database over attributes $\{x_i \mid r_i\in R\} \cup \{n\}\cup \{u_a,v_a \mid a\in A\}$. The tuples of $T$ are specified as before. That is, (1) for each $\dbtuple{a}\in T$, $\dbtuple{a}[n] = a$,
	(2) for each $r_i = (a,b)\in R$: \quad $\dbtuple{a}[x_i]=r_i,\quad \dbtuple{b}[x_i] = r_i,$ 
	\begin{align*}
		& \dbtuple{a}[u_b] = b, 
		& \dbtuple{a}[v_b] = b, \\
		& \dbtuple{b}[u_a]= a, 
		& \dbtuple{b}[v_a]= 0,
	\end{align*} 
	and (3) for each $\dbtuple{c} \in T$ and variables $x_i,u_d, v_d $ not already assigned: $\dbtuple{c}[x_i]  = c$,  $\dbtuple{c} [u_d] = 0$, and $\dbtuple{c}[v_d]= 0$.
	
\end{itemize}
Although for better presentation we used an attribute and FD for each $r\in R$, Lemma~\ref{lem:few-FDs} is still applicable and allows us to utilize $\gamma+1$ many FDs when $\calF$ has degree $\gamma$. 
Further, the encoding of conflicts via FDs is still symmetric, i.e., an attack of the form $r_i = \{a,b\}$ is considered for assigning attribute values for $x_i$'s.
The following example demonstrates the inconsistent database 
from our running example (Example~\ref{intro:ex-AF}). 

\begin{example}\label{ex:AF-both}
	The AF-database for $\calF$ is $\AF{A}{F}=\langle T,D\rangle$, where $D=F\cup I$ with $F=\{\depas{x_i}{n} \mid i\leq 3\}$ and $I=\{\incas{u_s}{v_s} \mid s\in A\}$, and $T$ is specified in Table~\ref{tab:AF-both}.
	\begin{table}[t]
		\centering

		\resizebox{.75\textwidth}{!}{
			\begin{tabular}{l|cccc|cccccccc}
				$T$ & $x_1$ & $x_2$ & $x_3$ & $n$ & $u_a$ & $v_a$ & $u_b$  & $v_b$ & $u_c$ & $v_c$ & $u_d$ & $v_d$\\
				\hline 
				$\dbtuple{a}$ &$r_1$ & $r_5$ &$r_2$ & $a$
				&$0$ & $0$ &$b$ & $b$ & $c$ & $c$ & $d$ & $d$\\
				$\dbtuple{b}$ &$r_1$ & $r_3$ &$b$ & $b$ 
				&$a$ & $a$ &$0$ & $0$ &$c$ &$c$ & $0$ & $0$\\
				$\dbtuple{c}$ & $r_4$ & $r_3$ &$r_2$ & $c$
				& $a$ & $0$ &$b$ & $0$ &$0$& $0$ & $d$& $d$\\
				$\dbtuple{d}$ & $r_4$ & $r_5$ &$d$ &$d$ 
				& $a$ & $0$ &$0$ &$0$ & $c$ & $c$ & $0$&  $0$ 
			\end{tabular}
		}
		\caption{The inconsistent database for the AF $\calF$ in Example~\ref{ex:AF-both}.}
		\label{tab:AF-both}
	\end{table}
\end{example}

\paragraph*{Self-attacking arguments.} 
A self-attacking argument can not belong to any extension. 
We add a single ID $i_s\dfn \incas{u_{s}}{v_s}$ over fresh attributes $\{u_s,v_s\}$ to encode  self-attacking arguments.
To this aim, for each $a\in A$ such that $(a,a)\in R$: we set $\dbtuple{a}[u_s]=a$ and $\dbtuple{a}[v_s]=0$ for the tuple $\dbtuple{a}\in T$. 
Since there is no tuple $\dbtuple{b}\in T$ with $\dbtuple{b}[v_s]=a$, the tuples corresponding to self-attacking arguments do not belong to any repair due to $i_s$.
Henceforth, we assume that AF-databases include the attributes $\{u_s,v_s\}$ and the ID $\incas{u_{s}}{v_s}$ for self-attacking arguments.
The equivalence between repairs of the conflict database and conflict-free (naive) extensions (Thm.~\ref{thm:conf-naive}) does not require the ID $i_s$ as the conflict-free extensions in an AF $\calF$ and those in the AF $\calF'$ obtained from $\calF$ by removing self-attacking arguments, coincide.

We are ready to prove that admissible and preferred extensions for $\calF$ coincide with repairs (max. repairs) for $\AF{A}{F}$.
%
\begin{theorem}\label{thm:adm-pref}
	Let $\calF$ be an AF and $\AF{A}{F}$ denotes its corresponding AF-database. Then, for every $S\subseteq A$, $S$ is admissible (resp., preferred) in $\calF$ iff $S_T \subseteq T$ is a repair (max. repair) for $\AF{A}{F}$.
\end{theorem}

\begin{proof}
	Let $\calF$ be an AF and $\AF{A}{F}$ denotes its corresponding AF-database.
	Recall that we defined $\AF{A}{F} =\langle T, D\rangle$, where $T$ is the database obtained by \emph{combining} the conflict and defense databases, and $D=F\cup I \cup \{i_s\}$ consists of a collection of FDs $D$, IDs $I$, and the dependency $i_s$ for self-attacking arguments (assuming that $s$ is a fresh name not in $A$). 
	Specifically, $\AF{A}{F}$ has the following components.
	Suppose $R'\subseteq R$ denotes attacks $(a,b)$ where $a\neq b$.
	\begin{itemize}
		%
		\item $F \dfn \{\depas{x_i}{n} \mid r_i\in R'\}$ and $I\dfn \{\incas{u_a}{v_a} \mid a\in A\}$ and $i_s \dfn \incas{u_s}{v_s}$.
		\item $T\dfn \{\dbtuple{a}\mid a\in A\}$ is a database over attributes $\{x_i \mid r_i\in R'\} \cup \{n\}\cup \{u_a,v_a \mid a\in A\} \cup \{u_s,v_s\}$. The tuples of $T$ are specified as follows: (1) for each $\dbtuple{a}\in T$, $\dbtuple{a}[n] = a$,
		(2) for each $r_i = (a,b)\in R'$ (i.e., $a\neq b$): 
		\begin{align*}
			& \dbtuple{a}[x_i]=r_i, 
			& \dbtuple{b}[x_i] = r_i,\\
			& \dbtuple{a}[u_b] = b, 
			& \dbtuple{a}[v_b] = b,\\
			& \dbtuple{b}[u_a]= a, 
			& \dbtuple{b}[v_a]= 0,
		\end{align*} 
		(3) for each $a\in A$ with $(a,a)\in R$, $\dbtuple a[u_s]=a$ and $\dbtuple a[v_s]=0$,
		and (4) for each $\dbtuple{c} \in T$ and variables $x_i,u_d, v_d, u_s,v_s $ not already assigned: $\dbtuple{c}[x_i]  = c$,  $\dbtuple{c} [u_d] = 0$, $\dbtuple{c}[v_d]= 0$, $\dbtuple c[u_s]=0$, and $\dbtuple c[v_s]=0$.
		
	\end{itemize}
	We first prove the correspondence between admissible extensions in $\calF$ and repairs for $\AF{A}{F}$.
	Let $S\subseteq A$ be admissible in $\calF$.
	Then, the corresponding set of tuples $S_T\subseteq T$ satisfies each FD in $F$ since $S$ is conflict-free~(Thm.~\ref{thm:conf-naive}).
	Further, note that $S$ does not contain any argument $a\in A$ for which $(a,a)\in R$. Therefore, $S_T$ also satisfies the ID $i_s$, as $\dbtuple a [u_s] = 0 =\dbtuple a [v_s]$ is true for every $\dbtuple a\in S_T$.
	Finally, $S_T$ satisfies each ID in $I$ due to Lemma~\ref{lem:defend-ids}, since $S$ defends each of its argument in $\calF$.
	As a consequence, $S_T$ is a repair for $\AF{A}{F}$.
	
	Conversely, assume $S_T$ is a repair for $\AF{A}{F}$.
	Then, $S_T$ can not contain any tuple $\dbtuple t \in T$ for which $\dbtuple t[u_s]\neq 0$ since $S_T$ would fail the ID $i_s$ otherwise.
	Therefore, the corresponding set $S$ of arguments does not contain any argument that attacks itself.
	Now, we can apply Theorem~\ref{thm:conf-naive} to conclude that $S$ is conflict-free.
	Moreover, the admissibility of $S$ follows from the fact that $S_T$ satisfies each ID in $I$ together with Lemma~\ref{lem:defend-ids}.
	This completes the proof in the converse direction.
	
	This establishes the proof for admissible extensions and repairs.
	Finally, the maximality of extensions in $\calF$ coincides with the maximality of repairs in $\AF{A}{F}$ because of the correspondence between arguments in $\calF$ and tuples in $\AF{A}{F}$.
	This concludes the proof for preferred extensions and subset-maximal repairs for the AF-database.
\end{proof}

%
Regarding the size of the AF-database and the time to construct it for a given AF, we have the following result.

\begin{theorem}\label{thm:AF-time}
	Let $\calF=(A,R)$ be an AF.
	Then, its corresponding AF-database $\AF{A}{F}$ can be constructed in polynomial time in the size of $\calF$.
	Additionally, $\AF{A}{F}$ has a table of size and $|A|\times 3(|A|+1)$ and uses (at most) $|A|$-many FDs and $(|A|+1)$ IDs, for an AF with $|A|$ arguments. 
\end{theorem}

\begin{proof}
	Let $\calF =(A,R)$ be an AF and $\AF{A}{F}=\langle T,F\rangle$ be the corresponding AF-database.
	Then, $T$ has exactly one tuple ${\dbtuple{a}}$ for each argument $a\in A$.
	Therefore, the number of rows (tuples) in $T$ is the same as $|A|$.
	Due to Theorem~\ref{thm:conf-naive-time}, we know that the table in the conflict database $\AF{C}{F}$ has the size $|A|\times (|A|+1)$, and there are $|A|$-many FDs.
	Moreover, due to Theorem~\ref{thm:def-time}, the defense database has the size $|A|\times 2|A|$, and there are $|A|$-many IDs
	Finally, there are two attributes $\{u_s,v_s\}$ and an ID $i_s$ to cover self-attacking arguments.
	This results in an AF-database of size $|A|\times 3(|A|+1)$ with at most $|A|$-many FDs, and $(|A|+1)$-many IDs.
	The claim that the database can be constructed in polynomial time follows from the proofs of Theorems~\ref{thm:conf-naive-time} and \ref{thm:def-time}.
	This completes the proof to our theorem.
\end{proof}

For symmetric AFs~\cite{Coste-MarquisDM05} (i.e., $(a,b)\in R$ iff $ (b,a)\in R$ for every $a,b\in A$), the stable, preferred and naive extensions coincide. 
As a result, one only needs the conflict database and FDs to establish the equivalence between repairs and extensions. 
\begin{corollary}
	Let $\calF$ be a symmetric AF 
	and $\AF{C}{F}$ its conflict database. Then, for $S\subseteq A$ and $\sigma\in \{\naive,\stab,\pref\}$, $S\in\sigma(\calF)$ iff $S_T \subseteq T$ is a subset-maximal repair for $\AF{C}{F}$.
\end{corollary}

We conclude this section by observing that one can not simulate AFs via inconsistent databases using only one type of dependencies (FDs or IDs) under the complexity theoretic assumption that the class $\Ptime$ is different from $\NP$.
This holds because the problem to decide  
the existence of a non-empty repair for a database instance comprising only FDs or IDs 
is in $\Ptime$~\cite{mahmood2024computing}, 
whereas the problem to decide whether an AF has a non-empty admissible-extension 
is $\NP$-complete. 
The stated claim holds due to 
Theorem~\ref{thm:adm-pref}.
\section{$\hspace{-.6em}$Attribute-Based Repairs $\hspace{-.1em}$\&$\hspace{-.1em}$ Other Semantics$\hspace{-1em}$}\label{sec:rem}

Besides subset maximality, we introduce maximal repairs following the principle of \emph{content preservation}~\cite{Wijsen:2003}.
The motivation behind these \emph{attribute-based repairs} is to adjust subset-maximality by retaining maximum possible values from $\dom(T)$  for certain attributes of $T$.
Let $\calD= \langle T,D\rangle $ where $T$ is a database with values $\dom(T)$ and $D$ is a set of dependencies.
Let $X$ be a sequence of attributes in $T$ and $P\subseteq T$, then $P[X]$ denotes the database obtained from $P$ by restricting the attributes to $X$ and $\dom(P[X])$ denotes the set of values in $P[X]$.
Given $\calD$ and a sequence $X$ of attributes, we say that $P\subseteq T$ is a \emph{maximally covering repair} for $\calD$ with respect to $X$ if $P\in \repairs(\calD)$ and there is no $P'\in \repairs(\calD)$ such that $\dom(P'[X])\supset \dom(P[X])$.
Moreover, we say that $P$ is a \emph{fully covering repair} for $\calD$ with respect to $X$ if $\dom(P[X])=\dom(T[X])$.
In other words, a maximal covering repair retains maximum possible values from $\dom(T[X])$, whereas a fully covering repair takes all the values from $\dom(T[X])$.
If $X$ contains all the attributes in $T$ then we simply speak of maximally (fully) covering repair for $\calD$ without specifying $X$.
For a database instance $\calD$ and attributes $X$, by $\repairsmaxcov(\calD,X)$ (resp., $\repairsfullcov(\calD,X)$) we denote the set of all maximally (fully) covering repairs for $\calD$ with respect to $X$.
%

\begin{example}\label{ex:att-covering-rep}
	In the database from Example~\ref{ex:intro-rep},
	$\{\dbtuple{s}_1,\dbtuple{s}_3,\dbtuple{s}_5\}$ and $\{\dbtuple{s}_2,\dbtuple{s}_3,\dbtuple{s}_5\}$ are both subset-maximal as well as maximally covering repairs for $\calD$ considering all attributes.
	Moreover, $\{\dbtuple{s}_1,\dbtuple{s}_4,\dbtuple{s}_5\}$ is subset-maximal but not maximally covering (since $\{\dbtuple{s}_1,\dbtuple{s}_3,\dbtuple{s}_5\}$ covers more values).
	Notice that there is no fully covering repair for $\calD$. 
	Nevertheless, 
	$\{\dbtuple{s}_1,\dbtuple{s}_3,\dbtuple{s}_5\}$ is a fully covering repair for $\calD$ with respect to attributes $\{\texttt{Tutor},\texttt{Time},\texttt{Room}\}$ whereas $\{\dbtuple{s}_2,\dbtuple{s}_3,\dbtuple{s}_5\}$ is not. 
	Further, there is no fully covering repair for $\calD$ w.r.t $\{\texttt{Course}\}$ or $\{\texttt{Advisor}\}$.
\end{example}

A maximally covering repair 
is also subset-maximal for any instance $\calD$, whereas, the reverse is not true in general. 
Moreover, a fully covering repair is also maximally covering, though it may not always exist.
Observe that, certain attribute may not be fully covered in any repair and requiring the repairs to fully cover the domain with respect to all the attributes may not always be achievable.
This motivates the need to focus on specific attributes in a database and include repairs fully covering the values in those attributes.

\subsection{Stable, Semi-stable, and Stage Semantics}\label{sec:sss}%
We prove that stable, semi-stable and stage extensions in AFs can also be seen as subset-repairs of the inconsistent DBs.
This is achieved by dropping subset-maximality and utilizing the covering semantics for repairs.
Notice that, the three semantics require maximality for the range $S^+$ for a set $S$ of arguments.
%
%
Recall that the attributes $\{v_x \mid x\in A\}$ in the defense database encode arguments attacking $x\in A$.
Together with the \emph{name} attribute ($n$) from the conflict database, we can encode the range of a set $S$ of arguments via an AF-database.
Let $X_r \dfn \{v_x \mid x\in A\}\cup \{n\}$ denote the \emph{range} attributes. 
The following lemma proves the relation between values taken by a set $S_T$ of tuples under attributes $X_r$ and the range $S^+$ for a set $S\subseteq A$ of arguments.
\begin{lemma}\label{lem:range}
	Let $\calF$ be an AF and $\AF{A}{F}$ denote its AF-database.
	Then, for every set $S\subseteq A$ of arguments and corresponding set $S_T\subseteq T$ of tuples, $S^+ = \dom(S_T[X_r])\setminus \{0\}$.
\end{lemma}

\begin{proof}
	Let $s\in S^+$. 
	If $s\in S$, then $s\in \dom(S_T[X_r]) $ since $\dbtuple{s}[n]=s$ and $\dbtuple s\in S_T$. 
	Otherwise, there exists $t\in S$ such that $(t,s)\in R$. 
	Then, $\dbtuple{t}[v_s]=s$ and therefore $s\in \dom(S_T[X_r])$ once again as $\dbtuple t\in S_T$.
	Conversely, suppose $s\not \in S^+$.
	Since $s\not \in S$, this implies that $s$ does not appear in the column $S_T[n]$.
	Further, since $s$ is not attacked by any argument in $S$, there is no tuple $\dbtuple{t}\in S_T$ such that $\dbtuple{t}[v_s]=s$ and therefore $s\not\in \dom(S_T[X_r])$.
	Finally, the element $0\in \dom(S_T[X_r])$ but it does not correspond to any argument and therefore removed to establish the claim.
\end{proof}


Now, we consider the AF-database for an AF~$\calF$ with range attributes $X_r$ 
to determine values in $\dom(T)$ covered by any repair.
Since stable and stage semantics only require conflict-freeness, 
we drop the admissibility and hence the set of IDs, except for the ID ($i_s$) for self-attacking arguments. 
\begin{theorem}\label{thm:stab-stag}
	Let $\calF$ be an AF and $\AF{A}{F}$ denotes its corresponding AF-database where $\AF{A}{F}=\langle T, F\cup I\rangle$. 
	Moreover, let $i_s$ be the ID for self-attacking arguments, $\AF{C}{F} = \langle T, F\cup\{i_s\}\rangle$,  and let $X_r = \{v_x \mid x\in A\}\cup \{n\}$.
	Then, for every $S\subseteq A$ and corresponding $S_T\subseteq T$,
	\begin{itemize}
		\item $S\in \stab(\calF)$ iff $S_T\in\repairsfullcov(\AF{C}{F}, X_r)$,
		\item $S\in \stag(\calF)$ iff $S_T\in\repairsmaxcov(\AF{C}{F},X_r)$
		\item $S\in \semistab(\calF)$ iff $S_T\in\repairsmaxcov(\AF{A}{F},X_r)$
	\end{itemize}
\end{theorem}

\begin{proof}
	
	We first prove the claim regarding stable semantics.
	Let $S$ be stable in $\calF$, then $S_T\in \repairs(\calC)$ since $S$ is conflict-free.
	Notice that $T[X_r]$ only contains argument names and the element $0$.
	Since $\dbtuple{s}[v_s]=0$ for any $\dbtuple{s}\in S_T$, therefore $0\in \dom(S_T[X_r])$.
	Further, since every $x\in A\setminus S$ is attacked by $S$, hence there is some $s\in S$ with $(s,x)\in R$.
	Consequently, $\dom(S_T[X_r]) = \dom(T[X_r])$ and $S_T\in \repairsfullcov(\calC,X_r)$. 
	Conversely, suppose $S\not\in \stab(\calF)$.
	If $S$ is not conflict-free then clearly $S_T\not\in \repairsfullcov(\calC,X_r)$ (due to Thm.~\ref{thm:conf-naive}).
	Whereas, if $S$ is conflict-free, then $S^+\neq A$ (otherwise $S$ is stable in $\calF$).
	This implies that there is some $x\in A\setminus S$ such that $(s,x)\not \in R$ for any $s\in S$.
	Then, $x\not \in \dom(S_T[X_r])$.
	This implies, $\dom(S_T[X_r])\neq \dom(T[X_r])$ since $x\in \dom(T[X_r])$.
	As a result, $S_T$ is not a fully covering repair for $\calC$ w.r.t. $X_r$ and this establishes the claim.
	
	Next, we prove the third claim, whereas, the second claim follows a similar argument.
	Let $S\in \semistab(\calF)$, then $S$ is admissible and there is no admissible $S'\subseteq A$ in $\calF$ with $S^+_R\subsetneq (S')^{+}_R$. 
	Suppose to the contrary that $S_T\not\in \repairsmaxcov(\calA,X_r)$.
	Since $S_T\in\repairs(\calA)$ (Thm.~\ref{thm:adm-pref}), there must be some $P\in \repairs(\calA)$ such that $\dom(P[X_r])\supset \dom(S_T[X_r])$.
	Let $P_\calF$ denote the set of arguments corresponding to tuples in $P$. 
	Then, this leads to a contradiction that $S\in \semistab (\calF)$ since $S$ and $P_\calF$ are both admissible in $\calF$ and $P_\calF^+\supset S^+$ due to Lemma~\ref{lem:range}.
	Conversely, suppose $S_T\in \repairsmaxcov(\calA,X_r)$, then $S_T\in\repairs(\calA)$ and hence $S$ is admissible in $\calF$.
	Now, the fact that no admissible $S'$ in $\calF$ exists with $(S')^+\supset S^+$ follows since $S_T$ is a maximally covering repair for $\AF{A}{F}$ w.r.t $X_r$ and  $\dom(S_T[X_r])$ is maximum with respect to all the repairs for $\AF{A}{F}$.
	This established the claim proof in reverse direction.
	
	The claim for the stage semantics follows from the above proof after replacing the admissibility by conflict-freeness.
	This completes the proof to our theorem.
\end{proof}

For stable and stage semantics, the size of the database can be reduced since there are no IDs involved (except $i_s$).

\begin{corollary}\label{cor:stab-stag-time}
	Let $\calF=(A,R)$ be an AF.
	Then, its corresponding AF-database $\AF{C}{F}$ for stable and stage semantics has a table of size $|A|\times (2|A|+3)$, and uses $|A|$-many FDs with a single ID.
\end{corollary}

\begin{proof}
	The claim follows easily from the proof of Theorem~\ref{thm:AF-time}, since we can drop the columns corresponding to attributes  $\{u_a\mid a\in A\}$ from the AF-database  for stable and stage semantics.
\end{proof}

\section{Discussion, Conclusion, and Future Work}
We presented a database view of Dung's theory of abstract argumentation. 
While our results tighten the connection between two domains by simulating Dung's argumentation frameworks by inconsistent databases, they also provide the exact expressivity of AFs via integrity constraints. 
Roughly speaking: Dung's argumentation frameworks can be equivalently seen as inconsistent databases where ICs include functional and inclusion dependencies.
This strong connection allows us to transfer further ideas from one domain to another. 
As already pointed out, a repairing semantics based on the idea of maximizing \emph{range} from AFs yields maximal (full) content preservation across certain attributes. 
We propose this new family of repairs as a topic of 
interest. 

\noindent\textbf{Complete and Grounded Semantics.}
We note that the complete and grounded semantics may not have a natural counterpart in the subset-repairs setting for databases when FDs and IDs are employed.
Although, we can 
introduce a notion of closed repairs utilizing the connection between defending arguments and supporting tuples (Lemma~\ref{lem:def-sup}) to simulate the closure properties required by these semantics, this may not correspond to a natural definition for repairs in the databases.
Then the question arises: are there classes of ICs that can express these two semantics for AFs?

\paragraph*{Outlook and Future Research Directions.}
Our work is in line with the existing literature connecting the two domains; seeing tuples in the database as arguments~\cite{bienvenu2020querying,mahmood2024computing} or constructing arguments from KBs~\cite{vesic2012,croitoru2013}.
By inspecting the expressivity of AFs to integrity constraints, we believe to have opened a new research dimension within argumentation.
This study can be extended by exploring the expressivity of extensions and generalizations of AFs, e.g., dialectical frameworks (ADFs)~\cite{BrewkaWoltran10}, set-based AFs~\cite{NielsenP06} and claim-centric AFs~\cite{DvorakW20} to name a few.
Moreover, in claim-centric AFs, the so-called \emph{claim-level} semantics requires maximizing accepted sets of claims rather than arguments --- which can be covered by maximally covering repairs for suitable attributes.
Further, the connection between probabilistic~\cite{li2011probabilistic} and preference-based~\cite{kaci2018preference} AFs and their database-counterpart~\cite{lian2010consistent,staworko2012prioritized} also seems promising. 

The well-known reasoning problems in AFs include \emph{credulous} (resp., \emph{skeptical}) reasoning, asking whether an argument belongs to some (every) extension.
However, the natural reasoning problem for databases includes consistent query answering (CQA) under \emph{brave} or \emph{cautious} semantics.
A (Boolean) query is entailed bravely (resp., cautiously) from an inconsistent database if it holds in some (every) repair.
It is interesting to explore DB-variants of credulous and skeptic reasoning. That is, given an argument $a$ in an AF~$\calF$ and semantics $\sigma$, construct a query $q^{a}_\sigma$ such that $a$ is credulously (skeptically) accepted in $\calF$ under semantics $\sigma$ iff the query $q^{a}_\sigma$ is entailed bravely (cautiously) from the database $\AF{D}{F}$ for $\calF$ under appropriate repairing semantics $\rho_\sigma$.
Yet another interesting direction for future work is to utilize recently developed (so-called) decomposition-guided~\cite{fichte2021decomposition,fichte2023quantitative,hecher2024quantitative} reductions 
for reasoning in DBs.
These reductions would allow to establish tight runtime bounds for CQA via the known bounds for reasoning in AFs.
Exploring these directions is left for future work.

Finally, a complexity analysis for repairs involving maximally or fully covering semantics is on our agenda. 
Some complexity lower bounds for the decision problems (involving FDs and IDs) already follow from the results in this paper.
We next aim to target the complexity of repair checking and CQA involving not only FDs and IDs, but also more general types of constraints such as \emph{equality} and \emph{tuple} generating dependencies. 
It is noteworthy that these definitions for repairs introduce a preference on sets of tuples, distinguishing them from research on prioritized databases~\cite{staworko2012prioritized,kimelfeld2017detecting,bienvenu2020querying}, which typically involves tuple preferences. 

%

\clearpage
\paragraph*{Acknowledgement.}
The work has received funding from 
the Austrian Science Fund (FWF), grants J 4656 and P 32830, 
the Deutsche Forschungsgemeinschaft (DFG, German Research Foundation), grants TRR 318/1 2021 – 438445824, the European Union's Horizon Europe research and innovation programme within project ENEXA (101070305), 
the Society for Research Funding in Lower Austria (GFF, Gesellschaft für Forschungsf\"orderung N\"O), grant ExzF-0004, 
as well as the Vienna Science and Technology Fund (WWTF), grant ICT19-065, and 
the Ministry of Culture and Science of North Rhine-Westphalia (MKW NRW) within project SAIL, grant NW21-059D.

\bibliography{arxiv}


\begin{comment}
\clearpage
\section*{Reproducibility Checklist}

\noindent
\textbf{This paper:}

\begin{itemize}
	\item Includes a conceptual outline and/or pseudocode description of AI methods introduced: Yes.
	\item Clearly delineates statements that are opinions, hypothesis, and speculation from objective facts and results: Yes.
	\item Provides well marked pedagogical references for less-familiare readers to gain background necessary to replicate the paper: Yes.
\end{itemize}

\noindent
\textbf{Does this paper make theoretical contributions?} Yes.

\begin{itemize}
	\item All assumptions and restrictions are stated clearly and formally: Yes.
	\item All novel claims are stated formally (e.g., in theorem statements): Yes.
	\item Proofs of all novel claims are included: Yes.
	\item Proof sketches or intuitions are given for complex and/or novel results: Yes.
	\item Appropriate citations to theoretical tools used are given: Yes.
	\item All theoretical claims are demonstrated empirically to hold: N/A.
	\item All experimental code used to eliminate or disprove claims is included: N/A.
\end{itemize}

\noindent
\textbf{Does this paper rely on one or more datasets?} No.

\noindent
\textbf{Does this paper include computational experiments?} No.
\end{document}